\documentclass[letterpaper,twocolumn,10pt]{article}
\usepackage{usenix}

\usepackage{amsmath}
\usepackage{xspace}
\usepackage{amssymb,amstext}
\usepackage{mathtools}
\usepackage{pgf}
\usepackage{makecell}

\usepackage{array}
\usepackage{xurl}
\usepackage{amsthm}
\newtheorem{theorem}{Theorem}
\usepackage{algorithm}
\usepackage[noend]{algpseudocode}
\usepackage{multirow}
\usepackage{booktabs}

\newcommand{\system}{Cloak\xspace}

\begin{document}

\title{Cloak: Heuristic ORAM Optimization Through Fixed Temporal Distribution}

\author{
{\rm Onur Eren Arpaci}\\
University of Waterloo
\and
{\rm Florian Kerschbaum}\\
University of Waterloo
\and
{\rm Sujaya Maiyya}\\
University of Waterloo
}

\maketitle

\begin{abstract}
Encrypted cloud storage can hide data contents but still leak sensitive information through access patterns. ORAM addresses this by hiding access patterns, but existing ORAM systems are too inefficient to deploy in practice. We present Cloak, an oblivious storage system that dramatically improves performance by leveraging a simple, widely observed property of real workloads: \emph{temporal locality}, where recently accessed items are more likely to be accessed again soon. Instead of trying to make server accesses look perfectly uniform, Cloak makes server traffic follow a fixed, “recentness-biased” pattern and then uses real queries to fill as much of that traffic as possible. When the workload exhibits temporal locality, Cloak achieves overheads as low as $1.1\times$ over a non-oblivious and unencrypted baseline. Importantly, this heuristic affects only performance, not security. We evaluate Cloak on Netflix click-stream and Ethereum transaction traces, achieving 165{,}000 and 157{,}000 operations per second, respectively, on a single machine.
\end{abstract}

% no keywords

\section{Introduction}

Cloud storage plays a critical role for many businesses by providing on-demand scalability 
to accommodate fluctuating traffic volumes and storage needs. This scalability can substantially reduce costs compared to traditional on-premises deployments, where system administrators must provision for peak load from the start. However, outsourcing storage 
in plaintext to a potentially untrusted cloud provider risks exposing sensitive information. 
While encrypting data acts as a first layer of defense, encryption alone is insufficient to
protect data privacy from powerful adversaries who can observe the accesses made to encrypted 
data. Specifically, access-pattern attacks~\cite{islam2012access, cash2015leakage, ioannis2020seal, grubbs2018pump, grubbs2019learning, gui2019encrypted, kellaris2016generic, kornaropoulos2019data, lacharite2018improved, simon2021hiding, simon2022ihop, poddar2020practical, yupeng2016all} exploit features such as query frequencies, correlations across queries, and partial auxiliary knowledge to infer plaintext values in encrypted databases.

Oblivious Random Access Memory (ORAM)~\cite{oram} is a cryptographic primitive designed to prevent such attacks by ensuring that the access pattern observable by the storage server is independent of the client's true queries. Works such as  ~\cite{oram, larsen2018yes, persiano2019lower} show that any ORAM scheme must incur at least $\mathcal{O}(\log N)$ overhead in computation and communication, where $N$ is the dataset size. Informally, this implies that for every real 
query issued by a client, the scheme must access $\log N$ elements from the remote storage 
to hide the real element requested by the client. 

While many constructions (fully or nearly) achieve the asymptotic bound~\cite{stefanov2018path, chakraborti2018concuroram, natacha2018obladi, dauterman2021snoopy, sujaya2022quoram, sahin2016taostore, stefanov2013oblivistore, asharov2020optorama, ren2015constants, chakraborti2019concuroram}, this lower bound is still prohibitively expensive as the state of the art implementations still incur between $40\times$ and $100\times$ slowdown compared to the insecure baseline~\cite{sahin2016taostore, dauterman2021snoopy}. These systems typically rely on either a trusted proxy or trusted execution environments (TEEs) between clients and the storage server to run protocol logic. Some proposals scale ORAM horizontally to achieve reasonable throughput and latency~\cite{dauterman2021snoopy, setayesh2025treebeard}, but this primarily trades performance cost for monetary cost. For example, Treebeard~\cite{setayesh2025treebeard} reports 150{,}000 operations per second, but requires 16 machines to do so. As a result, ORAM continues to face a significant cost barrier that limits real-world adoption.

A separate line of work~\cite{grubbs2020pancake, maiyya2023waffle} explores weakening ORAM-style security in exchange for better performance. Given the limited deployment of traditional ORAM, this direction is pragmatic: partial protection can be preferable to none. Pancake~\cite{grubbs2020pancake} uses frequency smoothing and selective duplication to make the server-observed access distribution uniform. However, their threat model requires that clients
draw queries independently from a known frequency distribution. Oya and Kerschbaum~\cite{simon2022ihop} show this assumption is often unrealistic and can be exploited 
on real workloads consisting of correlated queries. Waffle~\cite{maiyya2023waffle} obfuscates 
access patterns by sending dummy requests to the least recently accessed elements, and exposes parameters that trade security for performance. However, the concrete security meaning of these parameters and, how exploitable the resulting leakage is in practice, is not clear to users.

To summarize, ORAM-based oblivious schemes preserve strong security guarantees but incur 
prohibitively high overheads. In contrast, recent works~\cite{grubbs2020pancake, 
maiyya2023waffle} achieve improved performance by weakening security, but either make strong 
assumptions about the input workload or have privacy leakages whose practical impact is 
unclear.

To address these shortcomings, we introduce \emph{\system}, a trusted proxy-based 
oblivious storage system that incurs overheads as low as $1.1\times$ relative to simple unencrypted 
storage on real-world traces, while leaking no information about the underlying access 
pattern. We make two key observations that explain why Cloak is feasible.

First, existing oblivious schemes typically aim to make adversary-observable accesses to the outsourced database \emph{uniform}. 
However, real-world workloads rarely access all database elements uniformly; instead, they exhibit temporal skew (or temporal locality), 
where recently accessed elements are more likely to be accessed again than elements that have 
remained untouched for long periods. This assumption is pervasive in computer systems, take 
for example caching, which is effective precisely because temporal locality holds. 
Intuitively, when client accesses are skewed, forcing a uniform server access 
distribution incurs high overhead, measured by the number of dummy queries per real query.

Second, we observe that uniformity is not the only way to achieve obliviousness. 
Obliviousness only requires that the server-observed access pattern be \emph{independent} 
of the true access pattern, implying that the observed pattern can follow any fixed 
distribution as long as it remains independent of the client access distribution.

Motivated by these observations, Cloak forces the server accesses to always
follow a \emph{temporally skewed} distribution. If and when client queries exhibit 
similar temporal skew, most server requests can be filled with real queries. If not, 
the proxy fills the remaining requests with dummy queries. Our goal is to improve the 
efficiency of oblivious access for real-world distributions by minimizing wasted
computation and bandwidth due to dummy queries. But \textit{importantly, even if the client access 
distribution deviates from this assumption, Cloak’s security guarantees remain unchanged}, unlike schemes like Pancake~\cite{grubbs2020pancake}.

\textit{\textbf{Overview:}}
Cloak generates temporally distributed server-side accesses by 
aggregating client requests and issuing batched queries to the server 
according to a deterministic batching rule. This batching algorithm 
must satisfy the following three requirements:
\begin{enumerate}
    \item Generate batches that conform to a specified, configurable temporal distribution.
    \item Maximize the fraction of real queries per batch when client access patterns closely match the configured distribution.
    \item Hide the true client access pattern from the server.
\end{enumerate}

To enforce a temporal pattern, we define \emph{reuse distance} of an stored element as the number of batches sent by the proxy since that element was last accessed. When the proxy includes an element in a batch, that element's reuse distance becomes 0, while the reuse distance of all other elements increases by 1. Cloak partitions all stored elements into disjoint \emph{reuse-distance sets} where elements within a set have the same reuse distance. When forming a batch, the proxy 
selects a fixed number of elements from each reuse-distance set; we call these fixed counts \emph{budgets}. For instance, if the budget for reuse distance 1 is 10 and the budget for reuse distance 2 is 5, when 
forming a batch, the proxy selects exactly 10 elements from the reuse-distance set 1, and exactly 5 elements from the set 2. These predetermined budgets, known to the adversary,
ensure that every batch follows the same (fixed) temporal access pattern.

The database admin can configure budgets per reuse-distance set, with each set having
a minimum budget of 1. This enables Cloak to meet the second requirement: if real client 
queries are temporally skewed in a way that matches the budgets, then Cloak fills most budget slots with real requests instead of dummies.

To satisfy the third requirement (hiding the true pattern), after each batch, the proxy re-encrypts the accessed values, shuffles the elements within the batch, and writes back the freshly encrypted, permuted batch to the server. This forms the most recent reuse-distance set. Executing the above steps hides which specific elements were accessed because, from the server's perspective, elements within the same reuse-distance set are indistinguishable. Consequently, the only information the server can observe is how many items were drawn from each reuse-distance set, which is dictated by the budgets, making the observed access pattern independent of the real one.

This design introduces two additional security considerations. First, because the real query distribution affects how many dummy requests are needed, it can also affect the number of batches sent to the server over a fixed time window. This creates \emph{volume/throughput leakage}: for example, unusually high traffic could indicate that client accesses are less temporally skewed (requiring more dummy filling). Attacks such as ~\cite{kellaris2016generic, gui2019encrypted, poddar2020practical, demertzis2020seal, grubbs2018pump, markatou2019full, kornaropoulos2021response, faber2015rich, falzon2022range, markatou2023attacks} exploit volume patterns to infer information about
private queries. Cloak mitigates the volume/throughput leakage by sending batches at fixed time intervals.  
To keep our scope focused, we adopt the simplest (but also most expensive) approach by fixing the batch rate to a constant. More efficient alternatives such as padding to the nearest power of two~\cite{ioannis2020seal} or adding differentially private noise can easily be integrated into Cloak's design.

Second, standard ORAM threat models often allow the adversary to act 
as a client: the adversary can inject queries, measure request-response 
times, and monitor traffic between clients and the proxy. 
We adopt a slightly restricted threat model in which the adversary may inject queries but \textit{cannot observe client-proxy traffic nor measure request-response times}.
This threat model is stronger than the adversaries considered in
\cite{grubbs2020pancake, maiyya2023waffle} who cannot inject queries. We argue that this
model matches many real deployments in which all users are authenticated and authorized before receiving any data. Examples include internal enterprise databases, government systems restricted to specific officials, hospital systems restricted to clinicians, and research databases restricted to approved lab members.

We implement Cloak and evaluate it on Netflix click-stream~\cite{follows2021netlix} and Ethereum transaction traces~\cite{day2018ethereum}. 
On these workloads, Cloak retains $91\%$--$94\%$ of the throughput of an unsafe, non-oblivious baseline, reaching up to 169{,}344 ops/s on a single machine. 
Compared to Treebeard~\cite{setayesh2025treebeard}, a tree-based ORAM system, Cloak achieves $2.8\times$--$7.6\times$ higher throughput and $2\times$--$5.5\times$ lower latency, 
and it outperforms Waffle~\cite{maiyya2023waffle} by $4.7\times$--$6.6\times$ in throughput despite Waffle providing strictly weaker security guarantees.

\section{Related Work}
\label{sec:related-work}

An ORAM~\cite{oram} protocol aims to hide a client’s access pattern from untrusted memory, ensuring that the server’s observed sequence of physical accesses is independent 
of the logical access sequence. Goldreich and Ostrovsky establish the $\mathcal{O}(\log N)$ asymptotic overhead for complete obliviousness, motivating a large body of work that 
optimizes bandwidth, latency, and client storage under the standard ORAM threat model
\cite{stefanov2018path, ren2015constants, chakraborti2018concuroram, sahin2016taostore,natacha2018obladi,sujaya2022quoram, setayesh2025treebeard, bindschaedler2015practicing}. 
Among the most influential practical designs are tree-based ORAMs, including Path ORAM \cite{stefanov2018path}, which popularized a simple binary-tree organization with a client-side stash and position map, and later optimizations such as Ring ORAM~\cite{ren2015constants} and Circuit ORAM~\cite{wang2015circuit} that reduce constants and/or improve the scheme's efficiency. More recent theoretical progress (e.g., optimal-amortized constructions~\cite{asharov2020optorama}) shows the limits of 
what can be achieved asymptotically, however their behavior in practice is yet to be
studied.

A parallel line of work integrates ORAM into storage systems and distributed datastores. 
ObliviStore~\cite{stefanov2013oblivistore} and Treebeard~\cite{setayesh2025treebeard} demonstrate that careful systems design (e.g., asynchronous I/O and distributed storage) can significantly improve throughput while retaining strong ORAM-style security guarantees. TaoStore~\cite{sahin2016taostore} formalizes and 
mitigates security pitfalls that arise under asynchronous scheduling and concurrent access, 
providing a non-blocking tree-based design suitable for proxy-mediated deployments. Subsequent work 
such as QuORAM \cite{sujaya2022quoram} brings replication and fault tolerance to ORAM-backed 
datastores, showing how to preserve obliviousness while providing quorum-style availability and 
consistency. However, as proven in~\cite{oram, larsen2018yes, persiano2019lower},
any ORAM scheme must incur at least $\mathcal{O}(\log N)$ overhead, which can be unacceptable to 
support the scale of existing applications that serve tens of millions of
requests~\cite{twitter,bronson2013tao,decandia2007dynamo,atikoglu2012workload}.

\begin{figure}[t]
    \includegraphics[width=1\columnwidth]{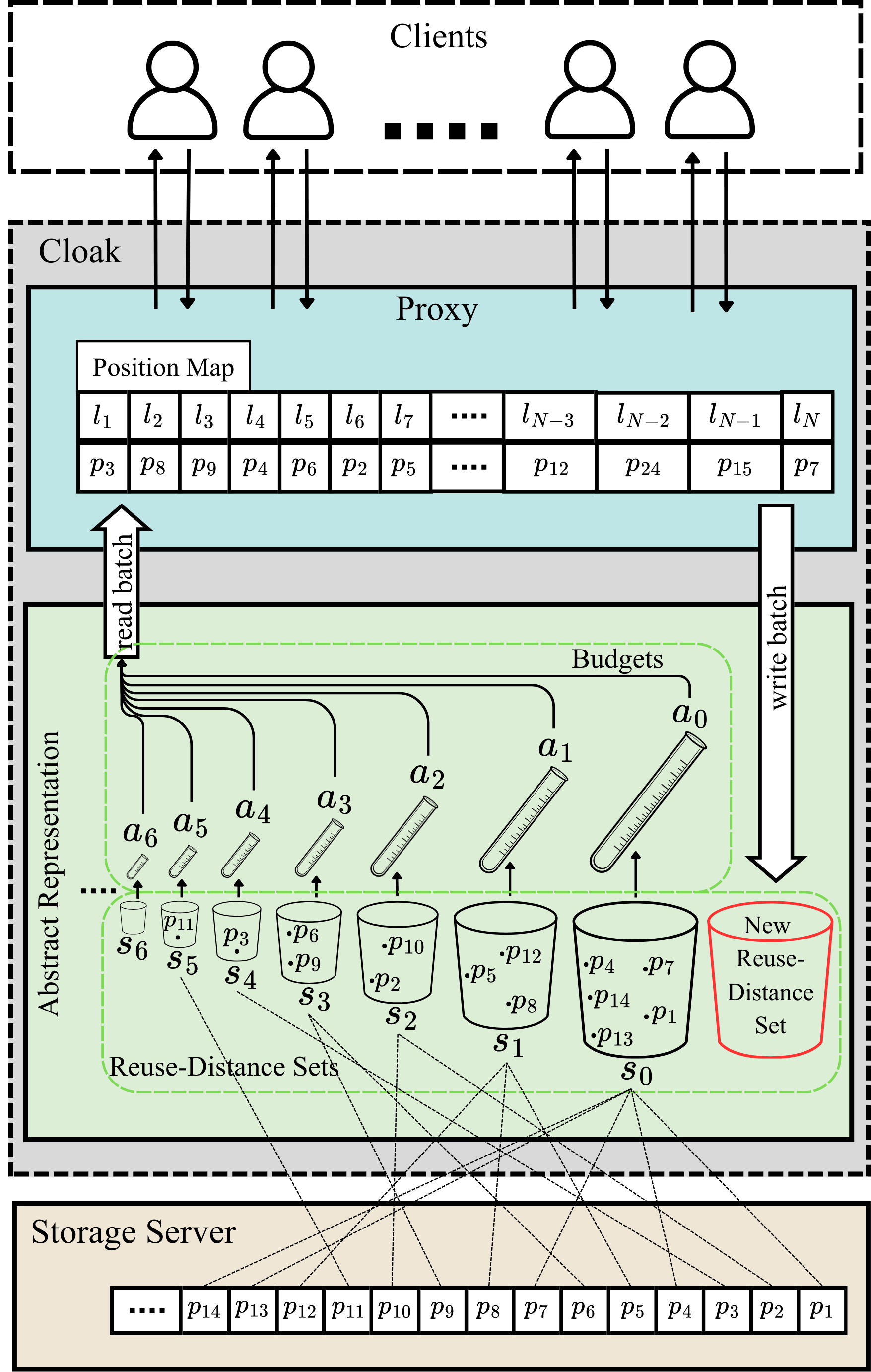}
    \caption{Cloak Overview}
    \label{fig:protocoloverview}
\end{figure}

Recent systems show that relaxing the adversarial model can unlock orders of magnitude 
higher performance than traditional ORAM. Pancake~\cite{grubbs2020pancake} introduces a 
passive persistent adversary model (the server observes accesses but cannot inject its own) and achieves small constant-factor overhead using \emph{frequency smoothing}. One of Pancake's main drawbacks is assuming that queries are selected independently from a frequency distribution, which ignores the natural correlations between queries. Using this vulnerability, IHOP~\cite{simon2022ihop} successfully attacks Pancake with real-world data. 
Moreover, Pancake's security holds as long as its anticipation of the client
query distribution matches the actual client query distribution. 

Waffle~\cite{maiyya2023waffle} removes Pancake's strict requirement of knowing the
input distribution a priori by providing \emph{online} protection under the same passive 
persistent model. It adapts to the observed sequence while maintaining constant 
bandwidth overhead and a bounded cache, and exposes tunable parameters that trade 
security for performance. However, the concrete security implications of poorly chosen 
parameter settings are unclear to users.

Cloak differentiates itself from Pancake~\cite{grubbs2020pancake} and Waffle~\cite{maiyya2023waffle} by decoupling performance from security. Both Pancake and Waffle trade-off security for performance, Pancake implicitly by assuming non-correlation between queries, and Waffle explicitly by providing parameters to tune security and performance. Whereas, Cloak leaks no information about the underlying query distribution. Unlike Pancake, a mismatch in anticipated query distribution can affect Cloak's performance but not its security.

\section{System and Threat Model}

\subsection{System Model} 

We model \system as a remote random access memory (RAM) supporting single address \texttt{READ} and \texttt{WRITE} operations.
Similar to commonly employed architecture in many oblivious systems \cite{stefanov2018path, natacha2018obladi, sujaya2022quoram, sahin2016taostore, grubbs2020pancake, maiyya2023waffle},
\system utilizes a trusted proxy to obliviously access encrypted and outsourced elements from a remote storage server. 
The proxy acts as a router between the clients and the storage server and executes the oblivious algorithm proposed in this work. 
Clients treat proxy as a simple remote RAM and send \texttt{READ} and \texttt{WRITE} queries to proxy. 
Similarly, the proxy treats the storage server as a remote RAM that accepts batch queries.
As part of the protocol, the proxy batches client queries before interacting with the storage server. 
Batches have a predetermined fixed size and are sent at fixed time intervals for security, both of which are configurable. 
While we assume the proxy to be always available, a real-world deployment can choose
techniques such as a primary-secondary replication~\cite{ghemawat2003google}
or a quorum replication~\cite{lamport2019part} to enforce high availability; however, this choice
is orthogonal to \system's goals.

\subsection{Threat Model}
Cloak protects organizational data from honest but curious adversaries who adhere to the protocol but may try to learn some private information about the data or the queries by persistently observing access patterns to the encrypted remote storage. 
Specifically, the adversary can monitor the accesses to the storage server over time and leverage these access traces to mount attacks such as frequency analysis, $l_p$-optimization, and quadratic query recovery 
\cite{lacharite2015note, li2017Information, naveed2015inference, simon2022ihop}. The adversary can also observe or delay (encrypted) messages between the proxy and the storage server and can inject queries. However, it cannot observe the responses to the 
injected queries nor observe or control the communication between clients and the proxy, 
including request-response times. Although weaker than traditional ORAM threat model
\cite{oram}, this threat model is significantly stronger than
than the \textit{snapshot} adversaries who can only access 
snapshots of the database without observing query accesses
~\cite{lacharite2017frequency, poddar2016arx, papadimitriou2016big}, and is stronger
than the passive persistent adversary~\cite{grubbs2020pancake, maiyya2023waffle}
who cannot inject queries.

\section{Protocol}
\label{sec:protocol}

\begin{table}[!tb]
\centering
\caption{Data Structure and Helper Function Definitions}
\begin{tabular}{>{\centering\arraybackslash}m{2.2cm}>{\arraybackslash}m{5.5cm}}
\toprule
\vspace{0.1cm} \textbf{Name} & \textbf{Definition} \\ \midrule[\heavyrulewidth]
$pmap$ & \vspace{0.1cm} Bi-directional map between logical and physical addresses. Each physical address also stores the most recent batch ID that accessed it. \vspace{0.1cm} \\ \hline
$reuse\_dist\_sets$  & \vspace{0.1cm} A list of sets with each set tracking the logical addresses with the same reuse distance. \vspace{0.1cm} \\ \hline
{\small $reuse\_dist\_queues$}  & \vspace{0.1cm} A list of queues per reuse-distance set holding pending (logical) addresses to be served.\\ \hline
Init & \vspace{0.1cm} Initializes $pmap$ and $reuse\_dist\_sets$ by (i) mapping each logical address to a random physical address, and (ii) randomly partitioning the logical addresses across reuse-distance sets so that $|s_t| = \sum_{i=t}^{len(budgets)} a_i$. \vspace{0.1cm} \\ \hline
ReadQuery & \vspace{0.1cm} Reads and parses a client query from $client\_conn$. \vspace{0.1cm} \\ \hline
SendResponse & \vspace{0.1cm} Sends a response value to a client over $client\_conn$. \vspace{0.1cm} \\ \hline
Forward & \vspace{0.1cm} Forwards a client query to the other coroutine via $query\_channel$. \vspace{0.1cm} \\ \hline
NextQuery & \vspace{0.1cm} Receives the next forwarded query from $query\_channel$. \vspace{0.1cm} \\ \hline
MapToPhysical & \vspace{0.1cm} Maps a list of logical addresses to their current physical addresses using $pmap$. \\ \hline
\makecell{Random-\\Permutation} & \vspace{0.1cm} Returns a random permutation over $batch\_size$ items. \vspace{0.1cm} \\ \hline
\makecell{StorageRead\\Write} & \vspace{0.1cm} Given a list of physical addresses, a permutation, and a list of \texttt{WRITE} queries: (i) reads the requested elements from the server, (ii) applies the writes, (iii) re-encrypts and permutes the elements, (iv) writes them back to the server, and (v) returns the decrypted read values. \vspace{0.1cm} \\ \bottomrule
\end{tabular}
\label{tab:definitions}
\end{table}

\subsection{Protocol Description}

This section describes \system's oblivious access protocol executed by
the client-side trusted proxy, with Figure~\ref{fig:protocoloverview} providing an overview.
The proxy routes client queries to the storage server in sequential \emph{batches}, each with a monotonically increasing id. 
A batch is a list of \emph{physical} addresses requested from the server; these addresses may correspond to real client queries or to dummy queries.

The proxy maintains a \emph{position map} (shown at the top of Figure~\ref{fig:protocoloverview}) that tracks the bidirectional mapping between logical addresses (used by clients) and physical addresses (used by the server), along with the batch id in which each physical address was last accessed.
The proxy also maintains a \emph{cache} of recently written values, allowing it to serve read queries for recently accessed data without waiting for the next batch.
The proxy operates on a fixed batching interval: it collects incoming client queries for $batch\_interval$ seconds, then sends a batched request to the server.
After receiving the server's response, the proxy serves the pending client queries, re-encrypts and shuffles the elements within the batch, updates the position map to reflect the shuffle, and writes the newly encrypted elements back to the server.

To enforce a target temporal access pattern, the proxy tracks the \emph{reuse distance} of the physical addresses.
Addresses with the same reuse distance form a reuse-distance set $s_t$, where $t$ denotes the reuse distance (depicted as buckets in Figure~\ref{fig:protocoloverview}). 
Each time the proxy sends a batch, it creates a new set $s_0$ such  that the id of every existing set increments by 1 (i.e., $s_t$ becomes $s_{t+1}$). 
The proxy assigns a fixed per-batch budget $a_t$ to each set $s_t$ (depicted as test tubes in Figure~\ref{fig:protocoloverview}) and in every batch, the proxy includes \emph{exactly} $a_t$ addresses from $s_t$. 
We require $a_t \ge 1$ for all $t$ to ensure that every portion of the database remains accessible.
 
Clearly, the notion of reuse-distance sets and the number of elements within
each set are known to the server (as the sets operate on physical addresses).
However, the encrypted elements within a reuse-distance set are indistinguishable 
to the server: after each batch, the proxy shuffles and re-encrypts the accessed elements before writing them back, and whenever an address is accessed,
it is removed from its previous reuse-distance set, so the remaining elements in that set stay indistinguishable. 
Because the proxy accesses a predetermined number of addresses (i.e., elements) from each reuse-distance set in every batch, the observable access pattern becomes fully deterministic. 
By allocating larger budgets to smaller reuse distances (more recent sets), we obtain low overhead on workloads with skewed temporal locality. 
In other words, if real queries have a similar temporal access pattern to the pattern implied by the budgets sizes, this batching algorithm ensures that most batch entries correspond to real queries, while hiding clients' access patterns.

\begin{algorithm}[t]
\caption{Get Client Queries}
\label{alg:getclientrequests}
\begin{algorithmic}[1]
\Procedure{GetClientQueries}{$client\_conn,$ $cache$, $query\_channel$}
    \Loop
        \State $query \gets$ ReadQuery($client\_conn$)
        \If{$query.type = \text{READ}$}
            \If{$cache.contains(query.address)$} \label{line:checkcache}
                \State \parbox[t]{0.9\linewidth}{SendResponse($client\_conn$, \\ \hspace*{2.5cm}$cache[query.address]$)}
            \Else
                \State Forward($query\_channel, query$)
            \EndIf
        \ElsIf{$query.type = \text{WRITE}$}
            \State $cache[query.address] \gets query.value$
            \State Forward($query\_channel, query$)
            \State SendResponse($client\_conn, \text{OK}$) \label{line:writecacheresponse}
        \EndIf
    \EndLoop
    \State $client\_conn.close()$
\EndProcedure
\end{algorithmic}
\end{algorithm}

The proxy runs two co-routines, \textsc{GetClientQueries} and \textsc{ProcessBatches}, to process client queries while concurrently interacting with the server.
Algorithms~\ref{alg:getclientrequests} and~\ref{alg:processrequets} provide the pseudocode and Table~\ref{tab:definitions} defines the data structures and helper functions.
\textsc{GetClientQueries} waits for a client query indefinitely. 
Upon receiving a \texttt{WRITE} query, it updates the cache, immediately acknowledges the client, and forwards the query to \textsc{ProcessBatches} via a buffered message channel.
Unlike systems like Pancake~\cite{grubbs2020pancake}, \system uses configurable fixed size cache with any eviction policy. Evicting dirty entries does not
impact correctness because each write is forwarded to \textsc{ProcessBatches}, which ensures it is eventually written to the server.
For a \texttt{READ} query, \textsc{GetClientQueries} first checks the cache and replies immediately on a hit; otherwise, it forwards the query to \textsc{ProcessBatches}.

\textsc{ProcessBatches} processes forwarded queries in a loop, tracking pending queries for each physical address in the $query\_map$ data structure (Algorithm~\ref{alg:processrequets}, line~4).
When a query arrives, if no other query for the same address is already pending, \textsc{ProcessBatches} enqueues it into the queue corresponding to its reuse-distance set (lines 8-11).

If another query exists for the same address, the algorithm handles the query without enqueuing it, for two reasons: (i)~each batch can include a given address at most once, so duplicate queries to the same address cannot be issued to the server, and 
(ii)~to minimize client latency, we respond to queries as soon as possible rather than deferring them to a future batch. Concretely, if the incoming query is a \texttt{READ} and another \texttt{READ} for the same address is already pending, the algorithm \emph{coalesces} them by appending the new request to the existing entry in $query\_map$ (line~\ref{line:coalesce}). All coalesced reads lead
to one physical read from the server and will be answered together once the address is fetched.
If the incoming \texttt{READ} finds a pending \texttt{WRITE} for the same address, the algorithm responds immediately using the buffered write value (line~\ref{line:writequeueresponsestart}).
Whereas, if the incoming query is a \texttt{WRITE}, it responds to all pending \texttt{READ}s for that address with the new value and replaces any previously pending \texttt{WRITE} (line~\ref{line:writequeueresponsestop}).
These optimizations do not violate linearizability, as we discuss formally
in Theorem~\ref{linearizability}.

If clients inject queries at a rate that exceeds the proxy’s processing capacity, requests will
accumulate in \texttt{query\_map}, leading to large queues. In extreme (and unlikely) cases, this
could cause excessive memory consumption at the proxy. To prevent this, \textsc{ProcessBatches}
bounds the size of \texttt{query\_map} and applies backpressure by blocking when the map is full,
causing the proxy to wait until an in-flight batch completes before accepting additional queries
(lines~20--21).

Independently, \textsc{ProcessBatches} checks whether $batch\_interval$ seconds have elapsed since the last batch was sent. If so, it constructs the next batch by popping up to $a_t$ addresses from the queues maintained for each reuse-distance. 
If a queue contains fewer than $a_t$ addresses, the proxy selects random addresses
for dummy queries from the corresponding reuse-distance set to fill the remaining budget.

\textsc{ProcessBatches} then sends a batched read to the server. 
Upon receiving the response, it replies to any pending client in $query\_map$,
applies any writes, re-encrypts and shuffles the elements in the batch, updates the position map, writes the permuted ciphertexts back to the server, and resets the batch timer. The addresses written back in the batch become reuse-distance set 0, and the process continues to handle the next set of incoming client queries.

\begin{algorithm*}[!tb]
\caption{Process Batches}
\label{alg:processrequets}
\begin{algorithmic}[1]
\Procedure{ProcessBatches}{$query\_channel, client\_conn$}
    \State $pmap$, $reuse\_dist\_sets \gets$ Init()
    \State $cur\_batch\_id \gets t-1$
    \State $query\_map \gets \emptyset$ \Comment{A map of address $\to$ pending reads or a pending write}
    \State $last\_batch\_time \gets time.now()$
    
    \While{receiving queries from $query\_channel$}
        \State $query \gets$ NextQuery($query\_channel$)
        
        \If{$query.address \notin query\_map$}
            \State $query\_reuse\_dist \gets cur\_batch\_id - pmap.lastBatchId[query.address]$
            \State $reuse\_dist\_queues[query\_reuse\_dist].insert(query.address)$
            \State $query\_map[query.address] \gets [query]$        
        \ElsIf{$query.type = \texttt{READ}$} \Comment{Coalesce reads / satisfy from pending write}
            \If{there is a pending \texttt{WRITE} in $query\_map[query.address]$} \label{line:writequeueresponsestart}
                \State SendResponse($client\_conn$, $query.id$, pending \texttt{WRITE} value)
            \Else
                \State $query\_map[query.key].append[query]$ \label{line:coalesce}
            \EndIf
        \ElsIf{$query.type = \texttt{WRITE}$}
            \State Respond to all pending \texttt{READ} queries in $query\_map[query.address]$ with $query.value$ \label{line:writequeueresponsestop}
            \State $query\_map[query.address] \gets [query]$
        \EndIf
        \If{$query\_map.len() \ge max\_queue\_capacity$}
            \State Sleep($batch\_interval - (time.now() - last\_batch\_time)$)
        \EndIf
        \State $is\_timeout \gets time.now() - last\_batch\_time > batch\_interval$
        \If {$is\_timeout$}
            \State $batch\_addresses \gets \emptyset$
            \State $batch\_write\_queries \gets \emptyset$
            \For{each $(reuse\_dist\_set, reuse\_dist\_queue)$ in $(reuse\_dist\_sets, reuse\_dist\_queues)$}
                \State $budget\_queries \gets$ Pop queries from reuse\_dist\_queue up to the budget or until the queue is emptied
                \State $budget\_queries \gets$ Add dummy queries from reuse\_dist\_set to fill remaining budget \label{line:adddummies}
                \State remove $budget\_queries$ from the reuse\_dist\_set
                \State $batch\_addresses.extend(budget\_queries.addresses)$
                \State $batch\_write\_queries.extend(budget\_queries.filter(type==\texttt{WRITE}))$
            \EndFor
    
            \State $physical\_addresses \gets$ MapToPhysical($pmap$, $batch\_addresses$)
            \State $permutation \gets$ RandomPermutation($|batch\_addresses|$)
            \State $responses \gets$ StorageReadWrite($physical\_addreses, permutation, batch\_write\_queries$)
            
            \State Send responses to clients for each \texttt{READ} query in $query\_map$ whose address was in $batch\_addresses$ \label{line:responsefromserver}
            \State Update $pmap$ with new physical addresses and $lastBatchId$
            \State $reuse\_dist\_sets.append(batch\_addresses)$
            \State $cur\_batch\_id \gets cur\_batch\_id + 1$
        \EndIf
    \EndWhile
\EndProcedure
\end{algorithmic}
\end{algorithm*}

\subsection{Dummy Selection Optimization}
\label{subsec:dummy-selection}

When constructing a batch, if the queue of pending real queries for a reuse-distance set $s_t$ contains fewer than $a_t$ entries, the proxy fills the remaining slots with dummy queries to randomly drawn addresses from $s_t$ (Algorithm~\ref{alg:processrequets}, line~\ref{line:adddummies}). Recall that addresses within the same reuse-distance set are indistinguishable to the server after shuffling and re-encryption. Therefore, the choice of \emph{which} addresses serve as dummies does not affect security, only the \emph{number} of addresses drawn from each set matters. This degree of freedom enables a performance optimization.

To implement this optimization, the proxy extends the position map to also track, for each logical address, the timestamp at which a client last issued a query for that address. Whenever a client query for an address arrives, this \emph{client-side access timestamp} is updated to the current time.

When selecting dummy queries from a reuse-distance set $s_t$, the proxy prefers addresses with the most recent client-side access timestamps. Formally, let $D_t \subseteq s_t$ denote the set of addresses available for dummy selection (i.e., hitherto unaccessed addresses in $s_t$). The proxy selects $k$ addresses from $D_t$ with the most recent client-side access timestamps, where $k = a_t - |Q_t|$ and $Q_t$ is the set of real queries targeting $s_t$ in the current batch.

\textbf{Rationale.}
This heuristic improves batch utilization under temporally skewed workloads. When an address for a dummy query is included in a batch, it moves to reuse-distance set $s_0$. Selecting recent client-accessed addresses as dummies preemptively repositions them into recent reuse-distance sets. Under temporal locality, recently accessed addresses are likely to be accessed again soon. When such an address is queried again, it will reside in a recent set $s_t$ with small $t$, which have larger budgets, improving the chances of serving
the query immediately in the subsequent batch. In effect, the dummy traffic anticipates future client demand, reducing the mismatch between real queries and available budget slots. Because addresses
within a reuse-distance set are indistinguishable, this optimization
is equivalent to choosing random addresses for dummy queries.

\section{Security and Correctness Analysis}
\label{sec:security}

\subsection{Security}
\label{sub:security}
In this section, we prove the security of Cloak in the semi-honest model. As mentioned in our threat model, the adversary can only observe the access traces to the remote storage server, so we can model the clients and the proxy as a single entity, and thus cache responses and query de-duplication do not bear any influence on the security proof. Since the batches are consistently sent at fixed time intervals, timing or amount of batches do not carry any private information. We assume the encryption scheme is IND-CPA secure so that ciphertexts are computationally indistinguishable from random numbers to any probabilistic polynomial-time (PPT) adversary, and we model the security of the protocol with the following game $\text{TRACE-IND}^{\Pi}_{\mathcal{A}}(\lambda)$:
\begin{itemize}
    \item The game chooses a challenge bit $b \leftarrow_{\text{\tiny{R}}} \{0,1\}$
    \item A PPT adversary chooses two arbitrary query sequences $Q_0$, $Q_1$ and sends them to the game.
    \item The game send $Q_b$ as the input queries to Cloak's proxy.
    \item The adversary starts observing batch traces sent to the storage server at each $batch\_interval$ seconds. Batch traces contain a list of physical addresses requested by the proxy and subsequently a list of encrypted values to be written to those physical addresses.
    \item Adversary can continue observing the batch traces up to $poly(\lambda)$ number of batches. (Protocol will continue sending batches even if $Q_b$ finishes)
    \item Finally, the adversary outputs a guess $b'$ for $b$. If $b'=b$, the adversary wins the game.
\end{itemize}
The adversary's advantage is defined as $Adv^{\Pi}_{\mathcal{A}} = |P(b'=b) - 1/2|$. We say that the protocol is secure if $Adv^{\Pi}_{\mathcal{A}}$ is negligible for every PPT adversary.

\begin{theorem}
For any PPT adversary $\mathcal{A}$ interacting in the game $\text{TRACE-IND}^{\text{Cloak}}_{\mathcal{A}}(\lambda)$, $\mathcal{A}$'s advantage is negligible assuming the underlying encryption scheme is IND-CPA secure.
\label{thm1}
\end{theorem}

\begin{proof}[Proof]
Let $H_0$ be the real game i.e., $\text{TRACE-IND}^{\text{Cloak}}_{\mathcal{A}}(\lambda)$. In $H_1$, we replace all encryptions and re-encryptions performed inside $StorageReadWrite$ with encryptions of a dummy value (e.g., zeros) of the same length. By IND-CPA security of the encryption scheme, any PPT adversary distinguishes $H_0$ and $H_1$ with only negligible advantage, so it suffices to argue indistinguishability in $H_1$.

Conditioned on $H_1$, the adversary only sees batch traces: the sizes of batches (fixed by design), the timing of batches (fixed by design), and the physical indices requested in each batch. Cache hits and de-duplication are invisible because the communication between the proxy and the clients is not visible to the adversary.

Fix any two query sequences $Q_0,Q_1$; we argue that the distribution of batch traces generated from $Q_0$ and $Q_1$ in $H_1$ are identical. The proof is done by induction over batches.
Let $pmap \colon \{l_1,l_2, \dots l_N\} \rightarrow \{(p_1,r_1), (p_2,r_2), \dots (p_N,r_N)\}$ be the position map, where $L \coloneq \{l_1,l_2, \dots l_N$\} are the logical addresses of the stored elements, $P \coloneq \{p_1,p_2, \dots p_N\}$ are the physical addresses of the stored elements in the storage server, and $R \coloneq \{r_1,r_2, \dots r_N\}$ are the batch index of the last batch that updated the corresponding physical address. Let $B^k \subset P$ denote the $k^{th}$ batch request by the proxy, and let $k_{cur}$ denote the current batch index. The storage is partitioned into disjoint reuse-distance sets. The reuse-distance set $s_t$ contains all physical addresses whose reuse distance is $t$, i.e., addresses last accessed $t$ batches ago: $s_t \coloneq \{ p_i \in P \mid k_{cur} - r_i = t \}$. The adversary can observe the batch requests $B^k \subset P$ sent to the server, and by tracking which addresses appear in each batch, it can derive all reuse-distance sets $s_t$.

\textbf{Base case:} Before the first batch, $pmap$ is created by randomly mapping each logical address $l_i$ to a physical address $p_j$ and assigning random batch indices that create reuse-distance sets with sizes dictated by the budgets. Let $B_0^0 \subset P$ and $B_1^0 \subset P$ be the first batches of $Q_0$ and $Q_1$ respectively. $|B_0^0| = |B_1^0|$ because of the fixed budget sizes. And due to the initial random mapping, $B_0^0$ and $B_1^0$ will be identically distributed. Therefore $Q_0$ and $Q_1$ cannot be distinguished by observing $B_0^0$ and $B_1^0$.

\textbf{Inductive step:} Assume batch traces up to batch $k-1$ are identically distributed. Let $B_0^k \subset P$ and $B_1^k \subset P$ be the $k^{th}$ batches of $Q_0$ and $Q_1$ respectively. Batches $B_0^k$ and $B_1^k$ will contain the same number of elements from each reuse-distance set, due to the predefined budgets. More formally: $|s_t \cap B_0^k| = |s_t \cap B_1^k| = a_t\ \forall{t}$. In addition, the distributions of $s_t \cap B_0^k$ and $s_t \cap B_1^k$ will be identical and independently uniformly random for all $t$. This holds because when a batch $B^k$ is processed, the logical addresses within that batch are randomly permuted before being written back, forming reuse-distance set $s_0$. As the current batch index advances, this set's reuse distance grows (from $s_0$ to $s_1$ to $s_2$, etc.), and each subsequent batch removes $a_t$ addresses from $s_t$. Since, by definition, a physical address is only accessed once while in a given reuse-distance set and then moved to $s_0$, all remaining elements in $s_t$ retain their fresh permutation. Thus, any equal-size subsets of $s_t$ are identically distributed. Therefore, $B_0^k$ and $B_1^k$ are identically distributed, and $Q_0$ and $Q_1$ cannot be distinguished by observing $B_0^k$ and $B_1^k$.

By induction, every batch trace is identically distributed across $Q_0$ and $Q_1$ in $H_1$, so the adversary's distinguishing advantage in $H_1$ is $0$. Combining with the IND-CPA hybrid transition yields only a negligible overall advantage in the real game. Therefore, the protocol is secure under the stated game.
\end{proof}

\subsection{Linearizability}

We define a history \textsf{Hist} as a sequence of logical read/write operations with each 
operation $op_i$ consisting of a pair of request-response events, \{\textsf{req$_i$, res$_i$}\}.
A \textit{complete} history is the one where every request event \textsf{req$_i$} in the history 
has a corresponding response event \textsf{res$_i$}; and otherwise
the history is said to be partial.

Let each read or write operation $op_i$ be represented as a tuple of $(l_i, v_i, u_i)$ where 
$l_i$ is the logical address of the element the operation targets, $v_i$ equals $\bot$ for
read operations and otherwise the element's value to be updated with, and $u_i$ is the response
value. Let $\leq_{lin}$ represent a linearizable relation between any two operations $op_i$ and
$op_j$: $op_i \leq_{lin} op_j$ implies \textsf{res}$_i$ precedes \textsf{req}$_j$ in a given 
history. Given a complete and finite history of operations executed by \system, this section 
proves \system is linearizable. 

Informally, linearizability implies that in the history, if an operation $op_i$ happens before
or precedes $op_j$, then $op_i$ should be ordered before $op_j$. This implies that if $op_j$ 
immediately succeeds $op_i$ and $op_i$ is a write, then $op_j$ should read the value written by 
$op_i$. Note that since linearizability is a per-element guarantee, $op_i$ and $op_j$ operate on 
the same logical element $l$. Formally, we say that the system is linearizable if there exists
a total order $\leq_{lin}$ over the operation identifiers such that: If \textsf{res$_i$} precedes 
\textsf{req$_j$} in the history such that no other write operation $op_k$ exists in between and 
if $v_i \neq \bot$, then $u_j = v_i$.

\begin{theorem}
\label{linearizability}
Cloak is linearizable.
\label{thm2}
\end{theorem}

\begin{proof}
    When $op_i$ precedes $op_j$ in the history and $op_i$ writes an element $l$ (i.e., $v_i \neq \bot$) with $op_j$ being the subsequent
    operation on $l$ with no other operation on $l$ in between, we show that $u_j = v_i$. To show this, we consider the following invariant:

    \textbf{Invariant.} A value $v_i$ written by $op_i$ will exist in at least one of the possible places: cache, queue, or server.

    This invariant holds because upon receiving the write operation $op_i$, the proxy adds $val_i$ 
    to the cache and forwards it to the \textsc{ProcessRequests} procedure where it gets queued. Since 
    \textsc{ProcessRequests} is single-threaded, only one process dequeues the requests atomically and 
    sends them to the server.  Note that we can ignore the time spent inside the $req\_channel$ as 
    the requests are sent sequentially by \textsc{GetClientRequests} and received sequentially by 
    \textsc{ProcessRequests}. Therefore, apart from possibly still exiting in the cache, $v_i$ will 
    either be in queue or written to the server. Given the invariant we examine each possible 
    situation with respect to the processing of the operation $op_j$:   
    \begin{itemize}
        \item \textbf{Stored in cache}. If $v_i$ is in cache, then $op_j$ will read $v_i$ from
        the cache as shown in Algorithm \ref{alg:getclientrequests}, line \ref{line:checkcache}.
        \item \textbf{Evicted from cache and waiting in queue}. If $v_i$ is already evicted from cache but still waiting in $query\_map$, then $op_j$ will get the response $v_i$ in Algorithm \ref{alg:processrequets}, lines \ref{line:writequeueresponsestart}--\ref{line:writequeueresponsestop}.
        \item \textbf{Evicted from cache, dequeued, and stored in the server}. If $v_i$ is only stored in the server, then $op_j$ will get the response $v_i$ in Algorithm \ref{alg:processrequets}, line \ref{line:responsefromserver} after the proxy retrieves
        the value from the server.
    \end{itemize} 
    Thus, $u_j = v_i$.
\end{proof}

\section{Budget Sizes for Reuse-distance Sets}
\label{sec:setting-budgets}

The protocol described in Section~\ref{sec:protocol} allows the database admin to choose an arbitrary 
budget distribution over reuse-distance sets (with a minimum budget of 1 for any set).
While the choice of budget distribution has no impact on the security of the protocol,
as shown in Theorem~\ref{thm1}, it critically impacts Cloak’s performance.
Cloak hides the true access pattern under a fixed target temporal access pattern observed by the
server. The closer the target distribution matches the client access distribution, the 
better its performance because \system can then populate each batch primarily with 
real queries rather than dummy queries. Since real-world workloads exhibit temporal skewness, i.e., recently 
accessed items are likely accessed again, it is natural to assign larger budgets to more 
recent reuse-distance sets. Determining an appropriate budget distribution, however, is non-trivial.

We hypothesize that Zipf’s law approximately characterizes the temporal locality of real-world access 
patterns. To evaluate this hypothesis, we analyze two datasets: Netflix click-stream traces \cite{follows2021netlix} and 
Ethereum transaction traces \cite{day2018ethereum}. In the Netflix dataset, each record corresponds to a user selecting a 
movie to watch and we interpret the timestamped events as a query trace over a movie database. The 
Ethereum dataset consists of all on-chain transactions from 2025-08-02 to 2025-08-13, in the execution order within the blockchain. We interpret these as queries to a database indexed by recipient 
wallet addresses.

\begin{figure}[t]
    \centering
    \input{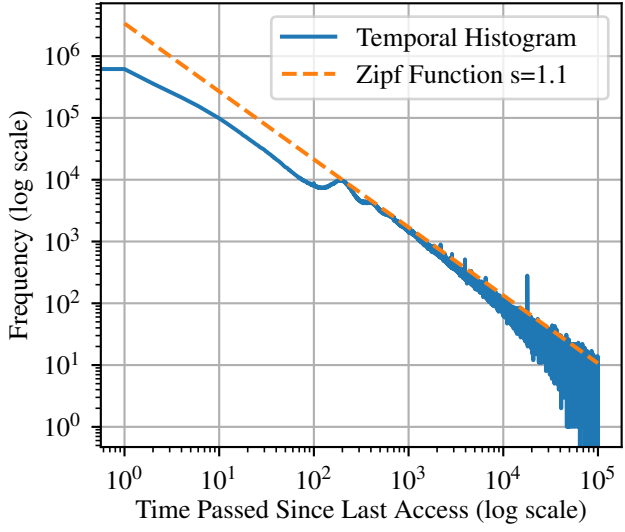}
    \caption{Ethereum Transaction Data Temporal Distribution (Log-Log Plot)}
    \label{fig:eth_dist}
\end{figure}

To characterize temporal locality, Figures~\ref{fig:eth_dist} and~\ref{fig:netflix_dist} plot, for 
each $x$, the number of access pairs to the same item that are separated by exactly $x$ intervening 
accesses, aggregated over all items. For example, at $x=10$, the $y$-axis counts the number of times an 
item is accessed and then re-accessed after 10 accesses to other items. We additionally plot the 
Zipf distribution with exponent $s$ that best fits each empirical distribution. Although neither 
dataset follows a strict Zipf distribution, both exhibit strong similarity. 
This empirical regularity is sufficient for Cloak to achieve high efficiency in practice.

Motivated by this observation, we compute budgets using Algorithm~\ref{alg:calc_budgets}. The algorithm 
takes the following inputs: the number of elements $N$ in the storage server and an initial budget $first\_budget$ for the most 
recent reuse-distance set. It initializes the budget list with $first\_budget$ and assigns subsequent 
budgets by decreasing $first\_budget$ according to a Zipf distribution with exponent $s=1.0$. The loop terminates once the 
resulting budget list is sufficient to cover the entire database, and the function returns the budget
list.

\begin{figure}[t]
    \centering
    \input{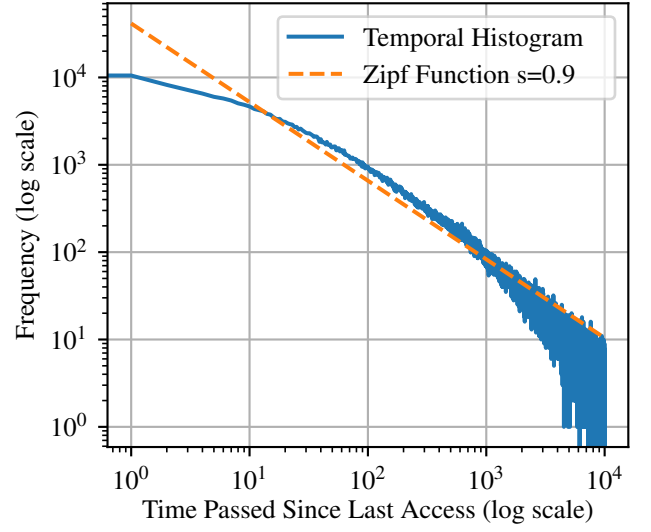}
    \caption{Netflix Click Stream Data Temporal Distribution (Log-Log Plot)}
    \label{fig:netflix_dist}
\end{figure}
The resulting budget list jointly determines the batch size and the sizes of all reuse-distance sets. 
Since each batch draws $budget[i]$ number of items from a reuse-distance set $i$, the batch size is
$|B| = \sum {budgets}$. Whereas, a reuse-distance set $s$ starts with $|B|$ entries, and every 
subsequent batch removes a decreasing yet fixed number of items from it, where the decrease is dictated
by the Zipf distribution. Thus, the size of set $s_t$ with reuse distance $t$ can be computed as $|s| = 
|B| - \sum_{i=0}^{t_{\mathrm{current}} - t} {budgets}[i]$. Because the database is partitioned into 
reuse-distance sets, the budget list also determines the total database size. Summing over all reuse-distance sets yields
\[
N = \sum |s_t| = \sum_{i=1}^{|budgets|} budgets[|budgets| - i + 1] \cdot i.
\]
From this relationship, we obtain the lower bound $|B| \ge \mathcal{O}(\sqrt{N})$. However, for 
temporally distributed client queries, much of this batch can be filled with real queries, as we
show in Section~\ref{sec:eval}. This constraint makes Cloak most suitable for high-throughput
workloads, where batching overhead can be amortized over many concurrent queries.

\begin{algorithm}
\caption{Set Budgets}
\label{alg:calc_budgets}
\begin{algorithmic}[1]
\Procedure{SetBudgets}{$N,$ $first\_budget$}
    \State $budgets \gets []$
    \State $total \gets 0$
    \State $i \gets 1$
    \While{$total < N$}
    \State $budget \gets \lceil first\_budget / i \rceil$
    \State $budgets.append(budget)$
    \State $total \gets total + (budget * i)$
    \State $i \gets i+1$
    \EndWhile
    \State return $budgets$
\EndProcedure
\end{algorithmic}
\end{algorithm}

\section{Evaluation}
\label{sec:eval}
\subsection{Setup}
We implement \system in Rust with $\sim1900$ lines of code. The system has 3 binaries, 
client, proxy, and server, that connect to each other with TCP connections. All
components are deployed on AWS across three m6a.4xlarge machines with 64GB RAM and
16 vCPUs. Cloak's implementation utilizes multiple threads to parallelize execution
where safe, with a default of 16 threads set for experiments. Each reported data point
in the experiments is an average of 5 runs to account for variations in a multi-tenant cloud.
We use OpenSSL Library for encryption and use AES-128-GCM as the encryption scheme.

\textbf{Baselines.} We compare \system with three baselines: (i) an unsafe baseline without any
encryption or obliviousness logic, where clients simply read and write elements from/to an in-memory
storage server. The unsafe baseline does not utilize batching or caching for query processing.
(ii) Treebeard~\cite{setayesh2025treebeard}, which
acts as a parallelized ORAM-based baseline. For a fair comparison with \system, we run Treebeard 
with a scale factor of 1 (i.e., no sharding for scalability), and deploy its storage server in
one machine, all three components of its proxy in a second machine, and the client on a third 
machine. (iii) Waffle~\cite{maiyya2023waffle}, which acts as an oblivious baseline with weaker
guarantees than \system (due to its more restrictive threat model). Similar to Treebeard, the
deployment places the storage server, proxy, and the client on three different machines.

\textbf{Datasets.} We compare all four systems on three datasets: Netflix \cite{follows2021netlix}, Ethereum \cite{day2018ethereum}, and synthetic datasets, with Table~\ref{tab:throughput-latency} describing the
dataset sizes. Section \ref{sec:setting-budgets} describes the details on the two real-world 
datasets.
The real-world datasets contain only read queries, while the synthetic dataset contains 50\% reads and 50\% writes.
To generate a temporally-distributed synthetic query trace, we repeatedly sample a value from a Zipf 
distribution and query the element whose current reuse distance equals the sampled value, updating all reuse 
distances after each query. This process continues until the trace reaches the desired length. 

\textbf{Configurable parameters.} Unless otherwise noted, the following are the default values
for \system's configurable parameters. The experiments set the cache size to 1000 elements,
the batching interval to 20 ms, and the batch size to 4{,}000 for all the experiments except 
the ones for the Netflix dataset, for which we set the cache size to 100, batching interval to 5 ms, and the batch size to 520 due to its small number of elements. The size of each element (in real or 
synthetic data) is set to 1KB. Because \system responds to queries, if possible, from the request
queue (Section~\ref{sec:protocol}), by default we set the queue size to be twice the batch size.
Finally, the budget size distribution in all experiments follow a discrete Zipf probability density function with the exponent $s=1.0$, calculated with the Algorithm \ref{alg:calc_budgets}. In other
words, \system anticipates the input distribution to follow Zipf $s=1.0$ by default and allocates
budgets to reuse-distance sets accordingly.

\begin{table}[t]
\renewcommand{\arraystretch}{1.2}
\centering
\caption{Throughput and latency across datasets and systems.}
\begin{tabular}{  w{c}{3.5em} w{c}{3.5em} w{l}{3.5em} m{4em} m{3.3em} }
\toprule
\textbf{Dataset} & \textbf{\#Elements} & \textbf{System} & \makecell{\textbf{Tpx}\\ \textbf{(ops/s)}} & \textbf{Latency (ms)} \\
\midrule[\heavyrulewidth]
\multirow{4}{*}{Netflix} & \multirow{4}{*}{8,473}   & unsafe     & 173,459 & 96 \\\cline{3-5}
        & & cloak      & 162{,}369 & 69 \\\cline{3-5}
        & & treebeard  & 58{,}823 & 166 \\\cline{3-5}
        & & waffle     & 34{,}415 & 29\footnotemark[1] \\
\midrule[\heavyrulewidth]
\multirow{4}{*}{Ethereum} & \multirow{4}{*}{701,314}  & unsafe     & 185{,}572 & 157 \\\cline{3-5}
        & & cloak      & 169{,}344 & 203 \\\cline{3-5}
        & & treebeard  & 23{,}935 & 416 \\\cline{3-5}
        & & waffle     & 25{,}650 & 39\footnotemark[1] \\
\midrule[\heavyrulewidth]
\multirow{4}{3.3em}{Synthetic Z.~exp:~1} & \multirow{4}{*}{1M} & unsafe     & 157{,}036 & 85 \\\cline{3-5}
        & & cloak      & 147,620 & 94 \\\cline{3-5}
        & & treebeard  & 19,370 & 520 \\\cline{3-5}
        & & waffle     & 24,994 & 40\footnotemark[1] \\
\bottomrule
\end{tabular}
\label{tab:throughput-latency}
\end{table}
\footnotetext[1]{Waffle uses an aggregate latency calculation method that produce very optimistic results.}

\subsection{Results}
\label{sec:results}
\subsubsection{Baseline Comparisons}

Table~\ref{tab:throughput-latency} summarizes the throughput and latency of all four systems across the 
three datasets. Cloak retains $91\%$--$94\%$ of the unsafe baseline's throughput on all three workloads, while providing formal obliviousness guarantees.
On the Netflix dataset, Cloak reaches 162{,}369 ops/s compared to the unsafe baseline's 173{,}459 ops/s, and even achieves lower latency (69\,ms vs.\ 96\,ms) due to caching effects. On the larger Ethereum and synthetic datasets, Cloak sustains 169{,}344 and 147{,}620 ops/s, respectively, with only modest latency increases over the unsafe baseline.

Compared to Treebeard, which provides ORAM-level security, Cloak delivers $2.8\times$--$7.6\times$ higher throughput and $2\times$--$5.5\times$ lower latency. Moreover, Treebeard's throughput degrades with dataset size, from 58{,}823 ops/s on Netflix ($N \approx 8{,}500$) to 19{,}370 ops/s on the synthetic dataset ($N = 10^6$), consistent with the $\mathcal{O}(\log N)$ overhead inherent in tree-based ORAMs. In contrast, Cloak's overhead relative to the unsafe baseline remains nearly constant across all three dataset sizes.

Cloak also outperforms Waffle by $4.7\times$--$6.6\times$ in throughput, despite Waffle providing strictly 
weaker security, as discussed in Section~\ref{sec:related-work}. The experiments set Waffle's batch size
to be same with \system, with up to 50\% real requests per batch, which is the medium security option 
suggested in~\cite{maiyya2023waffle}. Waffle's performance is significantly
worse that Cloak because (i) although Waffle utilizes a cache to store frequently queries elements, it 
does not respond to clients until after communicating with the server, and (ii) its batching logic is not
optimized for temporal queries, injecting a constant rate of dummy queries in every batch.
We note that Waffle's reported latencies are computed using an aggregate method that produces optimistic estimates, making a direct latency comparison misleading.

\subsubsection{Ablation Experiments}
Having established Cloak's performance with the state of the art baselines, the next set of experiments
analyze how different parameter selection affects \system's performance.

\begin{figure}[t]
    \centering
    \input{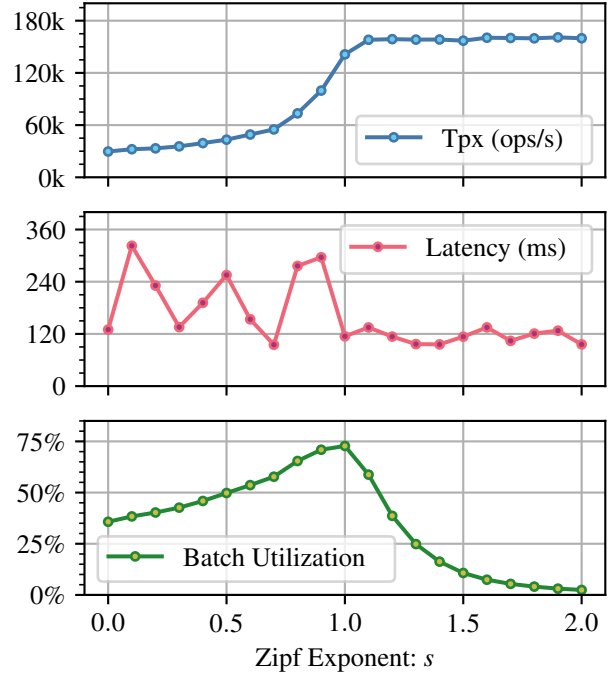}
    \caption{Throughput, mean latency, and mean batch utilization for different Zipf distributions ($N = 10^6$, Cache Size~$= 1{,}000$, Batch Size: $= 4{,}000$, Element Size $= 1$KB, Proxy Queue Capacity $= 2\times \text{Batch Size}$)}
    \label{fig:zipf-exponent}
\end{figure}

\textbf{Effect of workload skew.} \system achieves optimal performance when its estimate of the input
distribution closely matches the actual client query distribution. To study this effect, we fix the
anticipated Zipf skew at $s=1.0$ and vary the input workload from Zipf $s=0.0$ (uniform) to $s=2.0$.
Figure~\ref{fig:zipf-exponent} reports three metrics: throughput, end-to-end mean query latency, and
batch utilization. Batch utilization measures \system’s ability to populate each batch (on average)
with real client requests, and hence, higher utilization is desirable.

Figure~\ref{fig:zipf-exponent} indicates that throughput increases monotonically with the Zipf exponent, 
rising from roughly 30{,}000 ops/s under a near-uniform workload ($s=0$) to approximately 150{,}000 ops/s once $s \ge 1.0$, where it plateaus. The batch utilization graph (bottom panel) corroborates
this behavior, where utilization peaks at roughly $73\%$ near $s = 1.0$. At this point, the workload’s
temporal skew closely matches Cloak’s target distribution, which determines the budget allocation,
allowing \system to fill most of a batch with real requests. The sharp increase in throughput between
$s=0.8$ and $s=1.0$ corresponds directly to this improved alignment between the workload distribution
and the configured budgets.

For lower exponents ($s < 0.5$), the mismatch between the near-uniform workload and \system’s skewed
budget allocation reduces batch utilization to approximately $35\%$, leading to lower throughput.
Somewhat counter-intuitively, this correlation weakens or drops for $s > 1.0$: although batch utilization
declines, throughput remains near its peak. This occurs because highly skewed workloads repeatedly
access a small set of items, many of which are served from the cache or from recent
reuse-distance sets. As a result, batches include more dummy requests (particularly for older
reuse-distance sets), reducing utilization, while throughput remains high because most real requests
target cached or recently accessed elements.

Latency mirrors these trends: for $s \ge 1.0$, mean latency stabilizes at roughly 100--120\,ms, whereas for lower exponents it is both higher and more variable (spiking above 300\,ms) due to queue buildup when low batch utilization limits the number of real queries served per batch. These results demonstrate that Cloak is most efficient and effective when the its anticipated temporality matches the real 
distribution.

\begin{figure}[t]
    \centering
    \input{figures/stats_batch_size.pgf}
    \caption{Throughput, mean latency, and mean batch utilization for different batch sizes ($N = 10^6$, Cache Size $= 1{,}000$, Zipf Exponent $= 1.0$, Element Size $= 1$KB, Proxy Queue Capacity $= 2\times \text{Batch Size}$)}
    \label{fig:batch-size}
\end{figure}

\textbf{Effect of batch size.} The previous experiment fixed a batch size and varied the input distribution, whereas this experiment reverses the order: 
it fixes the input to Zipf distribution with exponent $s=1.0$ and studies how \system performs as the batch size changes. 
Figure~\ref{fig:batch-size} depicts the three performance metrics for this experiment. Throughput peaks at roughly 145{,}000 ops/s around a batch size of 4{,}000 and declines for both smaller and larger batches. Batch utilization follows a similar pattern, peaking at approximately $74\%$ near batch size 4{,}000 and decreasing steadily for larger batches.

Small batches limit throughput because fewer real queries can be served per batch, even though individual batches are processed quickly (latency is lower for smaller batches). Increasing the batch size initially improves throughput by providing more budget slots that can be filled with real queries. Beyond the optimum, however, the fixed 20 ms batching interval becomes a bottleneck: larger batches 
require more time for encryption, shuffling, and network transfer, and there aren't enough client 
queries accumulated within the batching interval to amortize this cost over. As a result, not only the 
batch utilization drops, but throughput also decreases due to more wasted compute and bandwidth overheads. Latency also rises---reaching roughly 140\,ms at a batch size of 11{,}000---as per-batch processing time grows. These results suggest that, for a given workload distribution, there exists a sweet spot for batch size that balances the number of available budget slots against the proxy's ability to fill and process them.

\begin{figure}[t]
    \centering
    \input{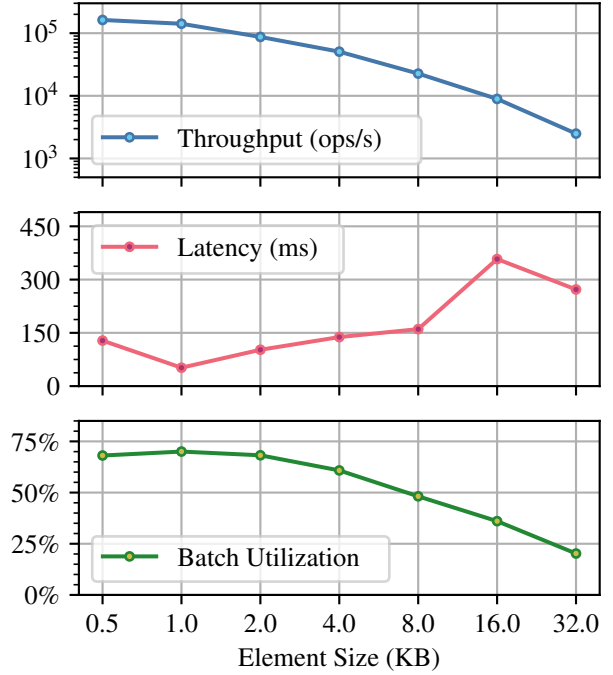}
    \caption{Throughput, mean latency, and mean batch utilization for different Element sizes ($N = 10^6$, Cache Size $= 1{,}000$, Zipf Exponent $= 1.0$, Batch Size: $= 4{,}000$, Proxy Queue Capacity $= 2\times \text{Batch Size}$)}
    \label{fig:element-size}
\end{figure}

\textbf{Effect of element size.} Next, we analyze the effect of element size variability on \system's
performance. Figure~\ref{fig:element-size} varies the element size from 0.5\,KB to 32\,KB on a log 
scale. Throughput decreases monotonically, from 162{,}425 ops/s at 0.5\,KB to 2{,}490 ops/s at 32\,KB, approximately halving with each doubling of element size. This is expected: since each batch transfers a fixed number of elements, the total data volume per batch scales linearly with element size, and the per-batch cost of encryption, shuffling, and network transfer grow accordingly. With the batching interval fixed at 20\,ms, larger elements reduce both throughput and batch utilization, with the latter dropping from $70\%$ at 1\,KB to $20\%$ at 32\,KB. Latency follows the inverse trend, rising from 51\,ms at 1\,KB to over 300\,ms at 16\,KB.

\begin{figure}[t]
    \centering
    \input{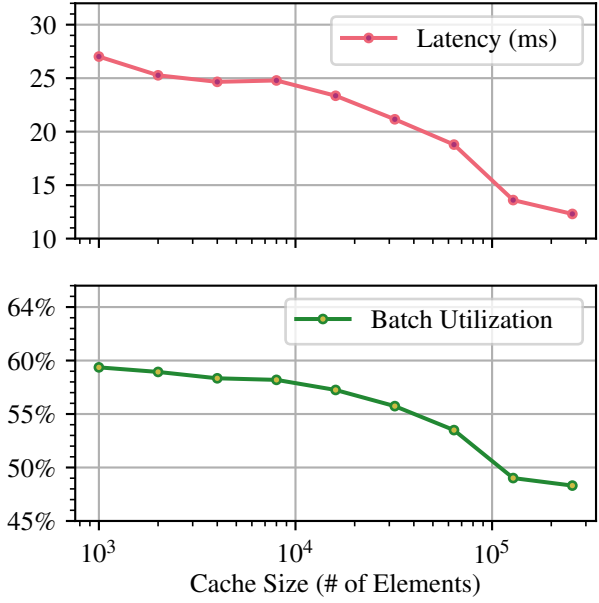}
    \caption{Mean latency, and mean batch utilization for different cache sizes ($N = 10^6$, Block Size $= 1$KB, Zipf Exponent~$= 1.0$, Batch Size: $= 4{,}000$, Client Query Rate $= 115{,}000$ ops/s, Proxy Queue Capacity $= 1\times \text{Batch Size}$)}
    \label{fig:cache-size}
\end{figure}

\textbf{Effect of cache size.} Finally, we study the impact of cache size on \system’s performance,
with a particular focus on latency. This analysis presents a subtle challenge: as described in
Section~\ref{sec:protocol}, the proxy stops accepting new client queries when its request queue
reaches capacity. When saturated, new requests are blocked, making it difficult to isolate the
effect of cache size. To avoid this challenge, we fix the client request rate below the
proxy’s maximum sustainable throughput, specifically at 115,000~ops/s.

Figure~\ref{fig:cache-size} reports mean latency and batch utilization as the cache size increases
(we omit throughput since it is fixed by the client at 115,000~ops/s). As the cache grows from
$1{,}000$ to $256{,}000$ entries, mean latency drops from 27 ms to 12.5 ms, because an increasing
fraction of queries are served directly from the cache without incurring a batch round-trip to the
server. Batch utilization decreases correspondingly, from $59\%$ to $48\%$, since fewer real queries
enter the batching pipeline. These results highlight the role of caching in reducing end-to-end
latency.

\section{Conclusion}

We presented Cloak ORAM, a trusted-proxy--based oblivious access protocol that achieves low overhead by 
decoupling obliviousness from uniformity. Instead of forcing server-side accesses to follow a uniform 
distribution, Cloak enforces a fixed, temporally skewed access pattern that better matches real-world 
workloads exhibiting temporal locality. This design preserves access-pattern independence while 
significantly reducing the number of dummy requests required in practice to achieve obliviousness. Our 
evaluation on real-world traces shows that Cloak can achieve overheads as low as $1.1\times$ relative to 
simple unencrypted storage. We believe Cloak provides a practical foundation for deploying access-pattern 
hiding in high-throughput settings and strikes a balance between cryptographic guarantees and realistic 
workload behavior. Cloak can be extended to support richer query workloads (e.g., range and join 
queries) and to dynamically adapt its batching interval to changing client input rate without compromising security using differential privacy.

\section*{Acknowledgments}
We gratefully acknowledge the support of National Cybersecurity Consortium, the Government of Canada (CSIN) (project id 2024-1345), NSERC grants RGPIN-2023-03244, IRC-537591, the Government of Ontario and the Royal Bank of Canada for funding this research.
\section*{Ethics Considerations.}
Our work used no sensitive, private, or personally identifiable data in designing or 
evaluating of Cloak. All experiments are conducted using synthetically generated workloads or publicly
available de-identified datasets. Because of this, when cloud resources were used for evaluation, only 
synthetic data or publicly available data were processed, ensuring that no sensitive information is 
disclosed to third-party service providers.

Cloak is designed to enhance access-pattern privacy by reducing information leakage through observable 
query timing and frequency. While these techniques are motivated by legitimate security and privacy 
goals, we acknowledge the potential for dual use. In particular, adversarial parties could deploy 
similar mechanisms to conceal undesirable or malicious behavior from monitoring systems. However, we 
believe that the benefits of advancing access-pattern privacy, especially for applications handling 
sensitive data, outweigh the potential for misuse.

Privacy-preserving data management systems play an important role in mitigating harms 
associated with data misuse and unintended information disclosure, which disproportionately affect users 
in sensitive domains such as healthcare, finance, and government services. By improving the efficiency 
and practicality of access-pattern hiding mechanisms, Cloak contributes to the broader goal of 
protecting user confidentiality in adversarial settings. We believe that the ethical implications 
of the project aligns with the principles outlined in the USENIX Security ethics guidelines.

\section*{Open Science}

We open-source all the implementation and experiments and are publicly available at: \url{https://github.com/onurerenarpaci/Cloak.git}

\bibliographystyle{plain}
\bibliography{references}

\end{document}